\documentclass[a4paper,12pt]{amsart}

\usepackage[utf8]{inputenc}
\usepackage[american]{babel}
\usepackage[hmarginratio={1:1}]{geometry}
\usepackage[shortlabels]{enumitem}
\usepackage{amsmath,amssymb}
\usepackage{amsthm}
\usepackage{comment}
\usepackage[hidelinks]{hyperref}
\usepackage{cleveref}

\newcommand{\ip}[2]{\langle#1,#2\rangle}
\newcommand{\abs}[1]{\lvert#1\rvert}
\newcommand{\norm}[1]{\lVert#1\rVert}

\newcommand{\id}{\mathrm{id}}
\newcommand{\OM}{\mathrm{OM}}

\renewcommand{\phi}{\varphi}

\newcommand{\dom}{\operatorname{dom}}

\newcommand{\IC}{\mathbb C}
\newcommand{\IR}{\mathbb R}

\theoremstyle{plain}
\newtheorem{proposition}{Proposition}[section]
\newtheorem{corollary}[proposition]{Corollary}
\newtheorem{lemma}[proposition]{Lemma}
\newtheorem{theorem}[proposition]{Theorem}

\theoremstyle{remark}
\newtheorem{remark}[proposition]{Remark}

\theoremstyle{definition}

\newtheorem{example}[proposition]{Example}

\title{The KMS and GNS Spectral Gap of Quantum Markov Semigroups}
\author{Melchior Wirth}
\address{Faculty of Mathematics and Computer Science, Leipzig University, Neues Augusteum, Augustusplatz 10, 04109 Leipzig, Germany}
\email{melchior.wirth@uni-leipzig.de}

\date{}

\begin{document}

\begin{abstract}
    We establish a relation between the exponential decay rates of quantum Markov semigroups with respect to different inner products. More precisely, it was conjectured by Fagnola, Poletti, Sasso and Umanità that for a Gaussian quantum Markov semigroup, the exponential decay rate with respect to the KMS inner product is bounded below by the exponential decay rate for the GNS inner product. We show that this is indeed the case and not limited to Gaussian quantum Markov semigroups, but holds for quantum Markov semigroups with a faithful normal invariant state on arbitrary von Neumann algebras. Additionally, the KMS inner product can be replaced by a whole class of inner products induced by operator monotone functions.
\end{abstract}

\maketitle

\section{Introduction}

Since their inception as models of the time evolution of certain open quantum systems in the seventies \cite{Lin76,GKS76,AL87}, quantum Markov semigroups and quantum Dirichlet forms have become a vital tool not only in mathematical physics, but also in adjacent fields of pure mathematics such as the theory of von Neumann algebras \cite{Pet09a,Pet09b,CS15,CS17}, noncommutative harmonic analysis \cite{MLX06,JM10,JM12,HJ24}, noncommutative geometry \cite{DR89,CS03,CS09} and noncommutative probability \cite{HP84,VS84,Bia03,CFK14}. We refer to \cite{Fag99,Cip08,Cip23} for some surveys.

Due to interactions with the environment, open quantum systems can show dissipative behavior that is absent in closed quantum systems. As a consequence, questions regarding the long-time behavior of quantum Markov semigroups have been subject to sustained research interest. Furthermore, questions of this kind are also of purely mathematical interest. For example, the existence of quantum Markov semigroups with certain dissipative features is tied to structural properties of the underlying von Neumann algebra -- one prominent example is Connes' spectral gap characterization of property Gamma for finite von Neumann algebras \cite{Con76}.

Mathematically, a quantum Markov semigroup is a family of unital completely positive maps on a von Neumann algebra that form a one-paramater semigroup with some continuity property. As such, multiples of the identity are fixed under a quantum Markov semigroup. Some typical questions regarding the long-time behavior ask whether these are the only fixed points, if the quantum Markov semigroup applied to an arbitrary observable converges to a fixed point as $t\to\infty$, and if so, how fast this convergence is.

Naturally, there are various variants of the last question depending on the measure of deviation between operators or states. In the classical case, if a Markov semigroup has an invariant probability measure and is reversible with respect to this measure, then exponential convergence in $L^2$ norm on the orthogonal complement of the constant functions is equivalent to the existence of a spectral gap of the $L^2$ generator, and the size of the spectral gap determines the exponential convergence rate.

A unique feature of the noncommutative setting is that for a given (faithful normal) invariant state, there is not only one inner product induced by the state, but a whole family of inner products. The most relevant examples in applications are the GNS (Gelfand--Naimark--Segal) and KMS (Kubo--Martin--Schwinger) inner product, but other examples such as the BKM (Bogoliubov--Kubo--Mori) have also been used. In fact, certain structural properties are probably best understood by studying the family of inner products parametrized by all operator monotone functions instead of only individual members.

Given this zoo of inner products, it is a natural question to ask how the exponential convergence rates for different inner products are related. We address this question in this article by proving a conjecture made by Fagnola, Poletti, Sasso and Umanità \cite{FPSU25} in the context of Gaussian quantum Markov semigroups.

Gaussian quantum Markov semigroups are quantum Markov semigroups on the bounded operators on Bosonic Fock space whose predual leaves the set of Gaussian states invariant. Gaussian quantum Markov semigroups are not only of high physical relevance, for example in quantum optics and continuous-variable quantum information, but also very interesting as mathematical models because many problems can be reduced to explicit finite-dimensional computations (in the finite-mode case). As such, Gaussian quantum Markov semigroup are a fascinating experimental ground for discovering general structural features of quantum Markov semigroups and have received intense research attention in recent years \cite{AFP21,AFP22a,AFP22b,Pol22,FP24,FL25,GP26}.

In the context of Gaussian quantum Markov semigroup, Fagnola, Poletti, Sasso and Umanità formulated the following conjecture \cite[Section 7]{FPSU25}: If a quantum Markov semigroup exhibits exponential decay with respect to the GNS inner product (on the orthogonal complement of the identity), it also exhibits exponential decay with respect to the KMS inner product and the exponential decay rate for the KMS inner product is at least as big as the one for the GNS inner product. They proved the conjecture for a class of Gaussian QMS in the single-mode case. Later, Fagnola and Li \cite{FL25} showed that the conjecture always holds in the single-mode case as long as the invariant Gaussian state has a density operator that is diagonal in the standard basis, and Li \cite{Li25} made progress on this conjecture by showing that exponential decay with respect to the GNS inner product implies exponential decay with respect to the KMS inner product for Gaussian quantum Markov semigroups.

In the present article, we solve the conjecture completely. In fact, we show that this relation between exponential decay rates with respect to the KMS and GNS inner product is not limited to Gaussian quantum Markov semigroups, but a general feature of quantum Markov semigroups with a faithful normal invariant state. Our proof relies on a by now classical circle of ideas, which has different incarnations in the form of the operator Jensen inequality \cite{HP82}, the transformer inequality for operator means \cite{KA80}, exact interpolation spaces between Hilbert spaces \cite{Don67} and in a special case also Lieb's concavity theorem \cite{Lie73}.

This article is structured as follows. In \Cref{sec:prelim} we collect some preliminaries on positive self-adjoint operators and quadratic forms, operator monotone functions and modular theory. In particular, we introduce for a given normal faithful state $\phi$ on  a von Neumann algebra a class of inner products $\langle\,\cdot\,,\cdot\,\rangle_f$ parametrized by the operator monotone functions $f\colon\IR_+\to\IR_+$.

In \Cref{sec:comparison} we prove an interpolation-type result (\Cref{thm:exact_interpolation}) for these inner products, which states that if a $\ast$-preserving map is contractive for GNS inner product, it is contractive for the $f$-inner product for all operator monotone $f$. As a consequence, a quantum Markov semigroup with invariant state $\phi$ extends to a strongly continuous contraction semigroup on the completion with respect to the $f$-inner product (\Cref{prop:C_0_ext}). Moreover, as another corollary of \Cref{thm:exact_interpolation} we obtain the comparison result between the exponential decay rates for the GNS and $f$-inner product that resolves the conjecture of Fagnola, Poletti, Sasso and Umanità (\Cref{cor:comparison_spectral_gap}).

\subsection*{Acknowledgements} The author is grateful to Franco Fagnola and Damiano Poletti for fruitful conversations on Gaussian quantum Markov semigroups and their spectral theory. This work was funded by the Deutsche Forschungsgemeinschaft (DFG, German Research Foundation) -- Project number 565405322.

\section{Preliminaries}\label{sec:prelim}

In this section we introduce some preliminary material regarding self-adjoint operators and quadratic forms, operator monotone functions and modular theory. All results here are either well-known to experts or follow easily from standard material.

\subsection{Positive self-adjoint operators and quadratic forms}

Let $H$ be a Hilbert space. A \emph{quadratic form} on $H$ is a map $Q\colon H\to [0,\infty]$ such that $Q(\lambda\xi)=\abs{\lambda}^2Q(\xi)$ and $Q(\xi+\eta)+Q(\xi-\eta)=2Q(\xi)+2Q(\eta)$ for all $\xi,\eta\in H$ and $\lambda\in \IC$. The \emph{domain} of $Q$ is given by $\dom(Q)=\{\xi\in H\mid Q(\xi)<\infty\}$. The form $Q$ is called \emph{densely defined} if $\dom(Q)$ is dense in $H$.

A quadratic form $Q$ is called \emph{closed} if it is lower semicontinuous and \emph{closable} if its restriction to $\dom(Q)$ is lower semicontinuous. If $Q$ is closable, then its lower semicontinuous envelope $\overline Q$ coincides with $Q$ on $\dom(Q)$ and is called the \emph{closure} of $Q$. In other words,
\begin{equation*}
    \overline Q(\xi)=\liminf_{n\to\xi}Q(\eta),\quad \xi\in H.
\end{equation*}

If $A$ is a positive self-adjoint operator in $H$ with spectral measure $E$, then 
\begin{equation*}
    Q_A\colon H\to [0,\infty],\,Q_A(\xi)=\int_{\IR_+}\lambda \,d\norm{E(\lambda)\xi}^2=\begin{cases}\norm{A^{1/2}\xi}^2&\text{if }\xi\in \dom(A^{1/2}),\\\infty&\text{otherwise}\end{cases}
\end{equation*}
is a densely defined closed quadratic form. Vice versa, if $Q$ is a densely defined closed quadratic form on $H$, then there exists a unique positive self-adjoint operator $A$ in $H$ such that $Q=Q_A$.

If $Q$ is a closed quadratic form and $\lambda>0$, then the \emph{Moreau regularization} $Q^{(\lambda)}$ is defined as
\begin{equation*}
    Q^{(\lambda)}\colon H\to [0,\infty),\,Q^{(\lambda)}(\xi)=\inf_{\eta\in H}(Q(\eta)+\frac 1 \lambda\norm{\xi-\eta}^2).
\end{equation*}
Lower semicontinuity immediately gives $Q^{(\lambda)}(\xi)\nearrow Q(\xi)$ for all $\xi\in H$.

The Moreau regularization is related to bounded approximations of positive self-adjoint operators as follows (see \cite[Section 1.~(i)]{Mos94} for example).

\begin{lemma}\label{lem:Moreau_approx}
    If $A$ is a positive self-adjoint operator in $H$ and $\lambda>0$, then $Q_A^{(\lambda)}=Q_{A(1+\lambda A)^{-1}}$.
\end{lemma}

\subsection{Operator monotone functions}

A function $f\colon \IR_+\to \IR_+$ is called \emph{operator monotone} if $f(A)\leq f(B)$ for all positive bounded operators $A$, $B$ on an arbitrary Hilbert space that satisfy $A\leq B$. We write $\OM$ for the set of all operator monotone functions from $\IR_+$ to itself and $\OM_1$ for the subset of all $f\in\OM$ with $f(1)=1$.

Operator monotone functions have been characterized by Löwner \cite{Low34} (see also \cite{Sim19} for a comprehensive overview with a variety of proofs). To state his result, let $t\geq 0$ and define 
\begin{equation*}
    h(t,\lambda)=\frac{(1+\lambda)t}{t+\lambda}
\end{equation*}
for $\lambda\in (0,\infty)$ and $h(t,0)=1$, $h(t,\infty)=t$. A function $f\colon\IR_+\to\IR_+$ belongs to $\OM$ if and only if there exists a finite Borel measure $m$ on $[0,\infty]$ such that
\begin{equation*}
    f(t)=\int_{[0,\infty]}h(t,\lambda)\,dm(\lambda),\quad t\geq 0.
\end{equation*}
Moreover, $f\in \OM_1$ if and only if $m$ is a probability measure.

Since we will deal with modular operators, which are typically unbounded, we need an extension of the Löwner order and operator monotone functions to unbounded operators. If $H$ is a Hilbert space and $A$,$B$ are positive self-adjoint operators in $H$, one defines $A\leq B$ by $Q_A(\xi)\leq Q_B(\xi)$ for all $\xi\in H$.

If $A$ is bounded, then $Q_A(\xi)=\langle \xi,A\xi\rangle$ for $\xi\in H$. Thus this notation is consistent with the usual Löwner order in the bounded case. In fact, the order for possibly unbounded positive self-adjoint operators can always be reduced to the bounded case in the following sense.

\begin{lemma}\label{lem:Lowner_bdd}
    Let $H$ be a Hilbert space. If $A,B$ are positive self-adjoint operators on $H$, then $A\leq B$ if and only if $A(1+\lambda A)^{-1}\leq B(1+\lambda B)^{-1}$ for all $\lambda>0$.
\end{lemma}
\begin{proof}
    We use that the quadratic form of $A(1+\lambda A)^{-1}$ has an explicit representation as Moreau regularization of $Q_A$ (see \Cref{lem:Moreau_approx}), namely,
    \begin{equation*}
        Q_{A(1+\lambda A)^{-1}}(\xi)=\inf_{\eta\in H}(Q_A(\eta)+\frac 1 \lambda\norm{\xi-\eta}^2),\quad \xi\in H.
    \end{equation*}
    An analogous identity holds for $Q_{B(1+\lambda B)^{-1}}$. From this identity it is clear that if $A\leq B$, then $A(1+\lambda A)^{-1}\leq B(1+\lambda B)^{-1}$ for all $\lambda>0$.
    
    Since $Q_A$ and $Q_B$ are lower continuous, we have $Q_{A(1+\lambda A)^{-1}}(\xi)\nearrow Q_A(\xi)$ and $Q_{B(1+\lambda B)^{-1}}(\xi)\nearrow Q_B(\xi)$ for all $\xi\in H$. Hence, if $A(1+\lambda A)^{-1}\leq B(1+\lambda B)^{-1}$ for all $\lambda>0$, then $A\leq B$.
\end{proof}

\begin{corollary}
    Let $H$ be a Hilbert space and $A,B$ positive self-adjoint operators on $H$. If $A\leq B$ and $f\in\OM$, then $f(A)\leq f(B)$.
\end{corollary}
\begin{proof}
    By Löwner's theorem, there exists a finite Borel measure $m$ on $[0,\infty]$ such that
    \begin{equation*}
        f(t)=\int_{[0,\infty]}h(t,\lambda)\,dm(\lambda).
    \end{equation*}
    Let $E_A$, $E_B$ denote the spectral measure of $A$ and $B$, respectively. By the spectral theorem,
    \begin{align*}
        Q_A(\xi)=\int_{[0,\infty]}\int_{\IR_+}h(t,\lambda)\,d\norm{E_A(t)\xi}^2\,dm(\lambda),\quad \xi\in H,
    \end{align*}
    and an analogous formula holds for $Q_B$ instead of $Q_A$. Hence it suffices to show that $h(A,\lambda)\leq h(B,\lambda)$ for all $\lambda\in [0,\infty]$. For $\lambda=0$ and $\lambda=\infty$, this is clear from the definition. For $\lambda\in (0,\infty)$, we have $h(A,\lambda)=\frac{1+\lambda}{\lambda}A(1+\lambda^{-1}A)^{-1}$ and an analogous identity holds for $B$ instead of $A$. Thus $h(A,\lambda)\leq h(B,\lambda)$ follows from \Cref{lem:Lowner_bdd}.
\end{proof}

Let us also note the following easy consequence of Löwner's theorem.

\begin{lemma}\label{lem:bound_OM}
    If $f\in\OM_1$, then $f(t)\leq t+1$ for all $t\geq 0$.
\end{lemma}
\begin{proof}
    Since for all $\lambda\in [0,\infty]$ one has $h(t,\lambda)\leq 1$ if $t\leq 1$ and $h(t,\lambda)\leq t$ if $t\geq 1$, this follows directly from Löwner's theorem.
\end{proof}

\subsection{Modular theory and inner products induced by a state}

Let $M$ be a von Neumann algebra and $\phi$ a faithful normal state on $M$. We denote by $(L^2(M,\phi),\pi_\phi,\Omega_\phi)$ a GNS triple associated with $\phi$, that is, $L^2(M,\phi)$ is a Hilbert space, $\pi_\phi\colon M\to \mathbb B(L^2(M,\phi))$ is a unital $\ast$-homomorphism and $\Omega_\phi\in L^2(M,\phi)$ such that $\phi(x)=\langle \Omega_\phi,\pi_\phi(x)\Omega_\phi\rangle$ for $x\in M$. Since $\phi$ is faithful, the representation $\pi_\phi$ is faithful. In the following, we will identify $M$ with $\pi_\phi(M)$ and omit writing $\pi_\phi$.

The map
\begin{equation*}
    M\Omega_\phi\to [0,\infty),\,x\Omega_\phi\mapsto \norm{x^\ast\Omega_\phi}^2
\end{equation*}
is a densely defined closable quadratic form. Let $Q$ denote its closure and let $\Delta_\phi$ be the unique positive self-adjoint operator on $L^2(M,\phi)$ such that $Q=Q_{\Delta_\phi}$. The operator $\Delta_\phi$ is called the modular operator. By definition, $M\Omega_\phi\subset \dom(Q)=\dom(\Delta_\phi^{1/2})$. Moreover, $\Delta_\phi$ is injective.

By Tomita's theorem, one has $\Delta_\phi^{it}M\Delta_\phi^{-it}=M$ for all $t\in \IR$. The family $\sigma^\phi=(\sigma^\phi_t)_{t\in \IR}$ given by $\sigma^\phi_t(x)=\Delta_\phi^{it}x\Delta_\phi^{-it}$ for $x\in M$ and $t\in\IR$ is called the \emph{modular group} of $\phi$.

\begin{lemma}
    If $f\in\OM_1$, then $f(\Delta_\phi)\leq \Delta_\phi+1$ and $M\Omega_\phi\subset \dom(f(\Delta_\phi)^{1/2})$.
\end{lemma}
\begin{proof}
    Let $\xi\in L^2(M,\phi)$ and let $E$ denote the spectral measure of $\Delta_\phi$. By \Cref{lem:bound_OM} we have
    \begin{equation*}
        Q_{f(\Delta_\phi)}(\xi)=\int_{\IR_+}f(t)\,d\norm{E(t)\xi}^2\leq \int_{\IR_+}(t+1)\,d\norm{E(t)\xi}^2=Q_{\Delta_\phi+1}(\xi).
    \end{equation*}
    Thus $f(\Delta_\phi)\leq \Delta_\phi+1$. By definition of $\Delta_\phi$ and the order on positive self-adjoint operators, $M\Omega_\phi\subset \dom(\Delta_\phi^{1/2})=\dom((\Delta_\phi+1)^{1/2})\subset \dom(f(\Delta_\phi)^{1/2})$.
\end{proof}

If $f\in \OM_1$, we define an inner product $\ip{\cdot}{\cdot}_f$ on $M$ by
\begin{equation*}
    \ip{x}{y}_f=\langle f(\Delta_\phi)^{1/2}x\Omega_\phi,f(\Delta_\phi)^{1/2}y\Omega_\phi\rangle,\quad x,y\in M,
\end{equation*}
and let $\norm\cdot_f$ denote the induced norm on $M$. We write $H_f$ for the completion of $M$ with respect to the inner product $\ip{\cdot}{\cdot}_f$. The previous lemma ensures that $\ip\cdot\cdot_f$ is well-defined.

This class of inner products has been studied before in the context of interpolation theory of Hilbert spaces \cite{Don67}, generalized quantum detailed balance conditions \cite{TKRWV10} and monotone metrics \cite{Pet96}.

\begin{remark}
    By definition, the norm $\norm\cdot_f$ is given by
    \begin{equation*}
        \norm{x}_f^2=Q_{f(\Delta_\phi)}(x\Omega_\phi),\quad x\in M.
    \end{equation*}
\end{remark}

\begin{example}
    \begin{enumerate}[(a)]
        \item If $f=1$, then $\ip\cdot\cdot_f$ is the usual GNS inner product, i.e., $\ip{x}{y}_1=\phi(x^\ast y)$, and $H_1$ is canonically identified with $L^2(M,\phi)$ via the map $x\mapsto x\Omega_\phi$. In this case, we also write $ \ip\cdot\cdot_{\mathrm{GNS}}$, $\norm\cdot_{\mathrm{GNS}}$ etc. instead of $\ip\cdot\cdot_{1}$, $\norm\cdot_{1}$ and so on.
        \item If $f=\id$, then $\ip{x}{y}_f=\phi(yx^\ast)$, a kind of ``anti-GNS'' inner product. As we will see below, these two inner products are extremal in a certain sense.
        \item If $f=\sqrt\cdot$, then $\ip\cdot\cdot_f$ is called the KMS inner product. In the case when $M=\mathbb B(H)$, $\phi=\operatorname{tr}(\,\cdot\,\rho)$ with a positive trace-class operator $\rho$, then the KMS inner product can be explicitly expressed as
        \begin{equation*}
            \ip{x}{y}_{\sqrt\cdot}=\operatorname{tr}(x^\ast\rho^{1/2}y\rho^{1/2}),\quad x,y\in \mathbb B(H).
        \end{equation*}
        As for the GNS inner product, we will also use the subscript $\mathrm{KMS}$ instead of $\sqrt\cdot$ for the KMS inner product, induced norm etc.
        \item If $f(t)=\frac{t-1}{\log t}$ for $t>0$ and $f(0)=0$, then the inner product $\ip\cdot\cdot_f$ is called the BKM inner product. It has recently found applications in the context of modified logarithmic Sobolev inequalities for quantum Markov semigroups and entropic gradient flows \cite{CM20,BM24}. If $M=\mathbb B(H)$, $\phi=\operatorname{tr}(\,\cdot\,\rho)$ with a positive trace-class operator $\rho$, then the BKM inner product can be expressed as
        \begin{equation*}
            \ip{x}{y}_f=\int_0^1 \operatorname{tr}(x^\ast\rho^s y\rho^{1-s})\,ds.
        \end{equation*}
    \end{enumerate}    
\end{example}

\section{Comparison of spectral gaps}\label{sec:comparison}

In this section, we establish that the KMS spectral gap is always bounded below by the GNS spectral gap (\Cref{cor:comparison_spectral_gap}). To do so, we first prove a general interpolation theorem (\Cref{thm:exact_interpolation}). This will also be instrumental in showing that a unital completely positive map with faithful normal invariant state is contractive with respect to the $f$-inner product from the previous section for an arbitrary operator monotone $f\colon\IR_+\to\IR_+$.

Throughout this section, let $M$ be a von Neumann algebra and $\phi$ a faithful normal state on $M$.

\begin{theorem}\label{thm:exact_interpolation}
    Let $P\colon M\to M$ be a linear projection that is $\ast$-preserving and satisfies 
    \begin{equation*}
        \phi(P(x)^\ast P(x))\leq \phi(x^\ast x),\quad x\in M.
    \end{equation*}
    Let $X=\operatorname{ran} P$. If $\Phi\colon X\to X$ is a $\ast$-preserving linear map such that 
    \begin{equation*}
        \phi(\Phi(x)^\ast \Phi(x))\leq \phi(x^\ast x),\quad x\in X,
    \end{equation*}
    then
    \begin{equation*}
        \norm{\Phi(x)}_f\leq \norm{x}_f,\quad x\in X,
    \end{equation*}
    for all $f\in\OM_1$.
\end{theorem}
\begin{proof}
    By assumption, $\norm{\Phi(x)\Omega_\phi}\leq \norm{x\Omega_\phi}$ for all $x\in X$. Thus $x\Omega_\phi\mapsto\Phi(x)\Omega_\phi$ extends to a contraction $C$ on $\overline{X\Omega_\phi}$. Moreover, since $\Phi$ is $\ast$-preserving, we have
    \begin{equation*}
        Q_{\Delta_\phi}(\Phi(x)\Omega_\phi)=\phi(\Phi(x)\Phi(x)^\ast)=\phi(\Phi(x^\ast)^\ast \Phi(x^\ast))\leq \phi(xx^\ast)=Q_{\Delta_\phi}(x\Omega_\phi)
    \end{equation*}
    for all $x\in X$.

    Now let us first consider the case $X=M$. If $\eta\in L^2(M,\phi)$, then by definition of $Q_{\Delta_\phi}$, there exists a sequence $(x_n)$ in $M$ such that $x_n\Omega_\phi\to \eta$ and $\limsup_{n\to\infty}Q_{\Delta_\phi}(x_n\Omega_\phi)\leq Q_{\Delta_\phi}(\eta)$. Thus
    \begin{equation*}
        Q_{\Delta_\phi}(C\eta)\leq\liminf_{n\to\infty}Q_{\Delta_\phi}(\Phi(x_n)\Omega_\phi)\leq \limsup_{n\to\infty}Q_{\Delta_\phi}(x_n\Omega_\phi)\leq Q_{\Delta_\phi}(\eta).
    \end{equation*}

    By definition, $\norm{\Phi(x)}_f^2=Q_{f(\Delta_\phi)}(\Phi(x)\Omega_\phi)$ and $\norm{x}_f^2=Q_{f(\Delta_\phi)}(x\Omega_\phi)$ for all $x\in M$. By Löwner's theorem, it suffices to consider the case $f(t)=t(1+\lambda t)^{-1}$ with $\lambda\in (0,\infty)$. For $x\in M$ we have
    \begin{align*}
         \norm{\Phi(x)\Omega_\phi}_f^2&=Q_{\Delta_\phi(1+\lambda \Delta_\phi)^{-1}}(Cx\Omega_\phi)\\
         &=\inf_{\eta\in L^2(M,\phi)}Q_{\Delta_\phi}(\eta)+\frac 1 \lambda\norm{\eta-Cx\Omega_\phi}\\
         &\leq \inf_{\eta\in L^2(M,\phi)}Q_{\Delta_\phi}(C\eta)+\frac 1 \lambda\norm{C(\eta-x\Omega_\phi)}\\
         &\leq \inf_{\eta\in L^2(M,\phi)}Q_{\Delta_\phi}(\eta)+\frac 1 \lambda\norm{\eta-x\Omega_\phi}\\
         &=Q_{\Delta_\phi(1+\lambda \Delta_\phi)^{-1}}(x\Omega_\phi)\\
         &=\norm{x\Omega_\phi}_f^2.
    \end{align*}
    This settles the case $X=M$. In the general case, define $\Psi(x)=\Phi(P(x))$ for $x\in M$. From the assumptions on $\Phi$ it follows easily that $\Psi$ is a $\ast$-preserving linear map and
    \begin{equation*}
        \phi(\Psi(x)^\ast\Psi(x))=\phi(\Phi(P(x))^\ast\Phi(P(x)))\leq \phi(P(x)^\ast P(x))\leq \phi(x^\ast x),\quad x\in M.
    \end{equation*}
    Thus we can apply our result in the case $X=M$ to $\Psi$ and obtain
    \begin{equation*}
        \norm{\Phi(y)}_f=\norm{\Psi(y)}_f\leq \norm{y}_f,\quad y\in X.\qedhere
    \end{equation*}
\end{proof}

\begin{remark}
    The assumptions on $\Phi$ imply that $\Phi$ is contractive with respect to both $\norm\cdot_1$ and $\norm\cdot_\id$. The previous theorem can then be seen as an instance of one direction of Donoghue's interpolation theorem \cite[Theorem 1]{Don67}, which states in our context that $H_f$ is an exact interpolation space for the pair $(H_1,H_\id)$.
\end{remark}

\begin{corollary}\label{lem:contractive_ext}
    Let $f\in \OM_1$. If $\Phi\colon M\to M$ is a unital completely positive map such that $\phi\circ\Phi=\phi$, then $\Phi$ extends to a contractive linear map $T$ on $H_f$.
\end{corollary}
\begin{proof}
    By the Kadison--Schwarz inequality,
        \begin{equation*}
            \phi(\Phi(x)^\ast\Phi(x))\leq \phi(\Phi(x^\ast x))=\phi(x^\ast x),\quad x\in M.
        \end{equation*}
    Thus the assumptions of \Cref{thm:exact_interpolation} are satisfied.
\end{proof}

\begin{corollary}\label{cor:interpolation}
    If $\Phi$ is unital completely positive map such that $\phi\circ \Phi=\phi$ and $\omega\in [0,1]$ such that
    \begin{equation*}
        \norm{\Phi(x)}_{\mathrm{GNS}}\leq \omega\norm{x}_{\mathrm{GNS}},\quad x\in M,\phi(x)=0,
    \end{equation*}
    then
    \begin{equation*}
        \norm{\Phi(x)}_{f}\leq \omega\norm{x}_{f},\quad x\in M,\phi(x)=0,
    \end{equation*}
    for all $f\in\OM_1$.
\end{corollary}
\begin{proof}
    This is trivial for $\omega=0$ and otherwise an immediate consequence of \Cref{thm:exact_interpolation} applied to $\frac 1\omega \Phi$, $P=\id_M-\phi(\cdot)1$ and $X=\ker \phi.$
\end{proof}

We want to apply this result to quantum Markov semigroups. Let us first state the definition. A family $(\Phi_t)_{t\geq 0}$ of normal unital completely positive maps on $M$ is called a \emph{quantum Markov semigroup} if $\Phi_0=\id_M$, $\Phi_{s+t}=\Phi_s\Phi_t$ for all $s,t\geq 0$ and $\Phi_t(x)\to x$ weak$^\ast$ as $t\to 0$ for every $x\in M$. We say that $\phi$ is an \emph{invariant state} of $(\Phi_t)_{t\geq 0}$ if $\phi\circ\Phi_t=\phi$ for all $t\geq 0$.

\begin{proposition}\label{prop:C_0_ext}
    If $(\Phi_t)_{t\geq 0}$ is a quantum Markov semigroup on $M$ with invariant state $\phi$ and $f\in\OM_1$, then $(\Phi_t)$ extends to a strongly continuous contraction semigroup $(T_t)_{t\geq 0}$ on $H_f$.
\end{proposition}
\begin{proof}
    By \Cref{lem:contractive_ext}, it only remains to show that $(T_t)_{t\geq 0}$ is strongly continuous. Since the family $(T_t)$ is uniformly bounded, it is enough to prove that $\norm{\Phi_t(x)-x}_f\to 0$ as $t\to 0$ for every $x\in M$. By \Cref{lem:bound_OM}, we have $\norm{\Phi_t(x)-x}_f^2\leq \norm{\Phi_t(x)-x}_{\id}^2+\norm{\Phi_t(x)-x}_{1}^2$. Thus it suffices to show the claim for $f=1$ and $f=\id$.

    By \cite[Theorem II.5.8]{EN00}, it is enough to show that $\langle (\Phi_t(x)-x)\Omega_\phi,y\Omega_\phi\rangle_f\to 0$ as $t\to 0$ for all $x,y\in M$. By definition, 
    \begin{align*}
        \langle (\Phi_t(x)-x)\Omega_\phi,y\Omega_\phi\rangle_1&=\phi((\Phi_t(x)-x)^\ast y),\\
        \langle (\Phi_t(x)-x)\Omega_\phi,y\Omega_\phi\rangle_\id&=\phi(y(\Phi_t(x)-x)^\ast).
    \end{align*}
    Both expressions converge to $0$ as $t\to 0$ as a consequence of the point-weak$^\ast$ continuity of $(\Phi_t)_{t\geq 0}$.
\end{proof}

Let $f\in\OM_1$. Following \cite{FPSU25}, we say that a quantum Markov semigroup $(\Phi_t)_{t\geq 0}$ on $M$ with invariant state $\phi$ has \emph{$f$-spectral gap} if there exists $\lambda>0$ such that
\begin{equation}\label{eq:spectral_gap}
    \norm{\Phi_t(x)}_f\leq e^{-\lambda t}\norm{x}_f,\quad x\in \ker\phi,\,t\geq 0.\tag{$\clubsuit$}
\end{equation}
In this case, we call the supremum of all $\lambda>0$ for which \eqref{eq:spectral_gap} holds the \emph{$f$-spectral gap} of $(\Phi_t)_{t\geq 0}$. As before, we use terms GNS and KMS spectral gap for $f=1$ and $f=\sqrt\cdot$, respectively.

Let us briefly clarify how this definition of the spectral gap is related to the numerical range of the generator. First recall that if $(T_t)_{t\geq 0}$ is a strongly continuous semigroup on a Banach space $X$, then its generator is the closed densely defined operator $A$ on $X$ given by
\begin{align*}
    \dom(A)&=\{x\in X\mid \lim_{t\to 0}\frac 1 t (T_t(x)-x)\text{ exists}\},\\
    Ax&=\lim_{t\to 0}\frac 1 t(T_t(x)-x),\quad x\in \dom(A).
\end{align*}

The following lemma is well-known in the theory of operator semigroups, we just include its short proof for the convenience of the reader.
\begin{lemma}\label{lem:exp_decay}
    If $H$ is a Hilbert space, $A$ the generator of a strongly continuous semigroup $(T_t)_{t\geq 0}$ on $H$ and $\lambda\in \IR$, the following two conditions are equivalent:
    \begin{enumerate}[(i)]
        \item $\norm{T_t}\leq e^{-\lambda t}$ for all $t\geq 0$,
        \item $\operatorname{Re}\langle\xi,A\xi\rangle\leq -\lambda\norm{\xi}^2$ for all $\xi\in \dom(A)$.
    \end{enumerate}
\end{lemma}
\begin{proof}
    (i)$\implies$(ii): If $\xi\in \dom(A)$, then
    \begin{equation*}
        \operatorname{Re}\langle\xi,A\xi\rangle=\lim_{t\to 0}\frac 1 t(\operatorname{Re}\langle \xi,T_t\xi\rangle-\norm{\xi}^2)\leq\liminf_{t\to 0}\frac {e^{-\lambda t}-1}{t}\norm{\xi}^2=-\lambda\norm{\xi}^2.
    \end{equation*}
    (ii)$\implies$(i): If $\xi\in\dom(A)$, then
    \begin{equation*}
        \frac{d}{dt}\norm{T_t\xi}^2=2\operatorname{Re}\langle T_t\xi,A T_t\xi\rangle\leq -2\lambda\norm{T_t\xi}^2
    \end{equation*}
    Thus $\norm{T_t\xi}^2\leq e^{-2\lambda t}\norm{\xi}^2$ by Grönwall's lemma. Since $\dom(A)\subset H$ is dense, we obtain $\norm{T_t}\leq e^{-\lambda t}$.
\end{proof}

Property (ii) from the previous lemma is essentially an upper bound on the spectrum of the real part of $A$. However, since $A+A^\ast$ might fail to be densely defined, it is in general not clear how to define the real part of $A$ (this problem does not occur for Gaussian quantum Markov semigroups at least when $f=1$ and $f=\sqrt\cdot$, as shown in \cite[Lemma 4.12]{FL25}).

Finally, the claimed relation between the GNS spectral gap and the $f$-spectral gap for $f\in\OM_1$ is a direct consequence of \Cref{cor:interpolation}.

\begin{corollary}\label{cor:comparison_spectral_gap}
    If $(\Phi_t)_{t\geq 0}$ has GNS spectral gap $\lambda>0$, then it has $f$-spectral gap at least $\lambda$ for all $f\in \OM_1$. In particular, it has KMS spectral gap at least $\lambda$.
\end{corollary}

\begin{remark}
    The computations from \cite[Section 7.1]{FPSU25} show that in general there is no inequality in the reverse direction, even up to multiplicative constants. In fact, it can happen that a quantum Markov semigroup has KMS spectral gap, but no GNS spectral gap.
\end{remark}

\begin{remark}
    The way the spectral gap is defined in \cite{FPSU25}, it implies that the ground state is non-degenerate, that is, multiples of $1$ are the only fixed points of $(\Phi_t)_{t\geq 0}$. More generally, one can also study the spectral gap in the case of a degenerate ground state. If $N=\{x\in M\mid \Phi_t(x)=x\text{ for all }t\geq 0\}$, then $N$ is a von Neumann subalgebra of $M$ and there exists a $\phi$-preserving conditional expectation $E\colon M\to N$ by \cite[Theorem A]{Wat79}. Then we may say that $(\Phi_t)_{t\geq 0}$ has $f$-spectral gap at least $\lambda>0$ for $f\in \OM_1$ if
    \begin{equation*}
        \norm{\Phi_t(x)}_f\leq e^{-\lambda t}\norm{x}_f,\quad x\in \ker E,\,t\geq 0.
    \end{equation*}
    In the light of \Cref{thm:exact_interpolation} (with $P=\id_M-E$), \Cref{cor:comparison_spectral_gap} remains valid in the case of a degenerate ground state.
\end{remark}

\begin{remark}
    In general, it seems difficult to determine the relation between the $f_1$-spectral gap and the $f_2$-spectral gap for $f_1,f_2\in\OM_1$. From the comparison between the GNS spectral gap and $f$-spectral, we can destillate the following two general facts: First, if $f\in\OM_1$ and $\tilde f$ is its transpose defined by $\tilde f(t)=tf(t^{-1})$ for $t\in\IR_+$, then
    \begin{equation*}
        \norm{x^\ast}_f=\norm{f(\Delta_\phi)^{1/2}x^\ast\Omega_\phi}=\norm{\Delta_\phi^{1/2}f(\Delta_\phi^{-1})^{1/2}x\Omega_\phi}=\norm{x}_{\tilde f},\quad x\in M.
    \end{equation*}
    As ucp maps are $\ast$-preserving, it follows that the $f$-spectral gap and the $\tilde f$-spectral gap coincide. In particular, if $f_\alpha(t)=t^\alpha$ for $\alpha\in [0,1]$, then $\tilde f(t)=t^{1-\alpha}$, and the $f_\alpha$-spectral gap coincides with the $f_{1-\alpha}$-spectral gap.

    Second, if $0\leq \alpha\leq \beta\leq 1/2$, then $\tilde f_\alpha(\Delta_\phi)=\Delta_\phi^{1-2\alpha}f(\Delta_\phi)$ and $f_\beta(\Delta_\phi)=(\Delta_\phi^{1-2\alpha})^{\frac{\beta-\alpha}{1-2\alpha}} f_\alpha(\Delta_\phi)$. Since $t\mapsto t^{\frac{\beta-\alpha}{1-2\alpha}}$ belongs to $\OM_1$, an analogous interpolation argument to that of \Cref{thm:exact_interpolation} shows that the $f_\alpha$-spectral gap is less than or equal to the $f_\beta$-spectral gap.

    Therefore, as a function of $\alpha$, the $f_\alpha$-spectral gap is symmetric around $1/2$, monotonically increasing on $[0,1/2]$ and monotonically decreasing on $[1/2,1]$.
\end{remark}

\bibliographystyle{alpha}
\bibliography{ref}

\end{document}